\documentclass[reprint,amsmath,amssymb,aps,pra,floatfix]{revtex4-2}
\usepackage{amsthm,booktabs,array}
\usepackage[hidelinks]{hyperref}
\newtheorem{theorem}{Theorem}
\newcommand{\R}{\mathbb R}\newcommand{\Z}{\mathbb Z}\newcommand{\C}{\mathbb C}
\newcommand{\tr}{\operatorname{tr}}
\newcommand{\norm}[1]{\left\lVert#1\right\rVert}
\usepackage{graphicx}
\graphicspath{{_work/}{./}}
\usepackage{dcolumn}
\usepackage{bm}

\begin{document}

\title{Approximation of a continuous-variable Gaussian state by discrete-variable quantum states}

\author{Nicolae Cotfas}
 \homepage{https://unibuc.ro/user/nicolae.cotfas/ \\
 Emails:  \ nicolae.cotfas@g.unibuc.ro \ ,\\ \ ncotfas@yahoo.com
 }
\affiliation{Faculty of Physics, University of Bucharest,\\
P.O. Box MG-11, 077125 Bucharest, Romania}%

\author{Haobo Ma}
\thanks{ORCID: 0009-0008-4944-977X}
\affiliation{ChronoAI Pte Ltd, Singapore}
\author{Wenlin Zhang}
\thanks{ORCID: 0009-0005-2414-7968}
\affiliation{National University of Singapore, Singapore}
\affiliation{ChronoAI Pte Ltd, Singapore}

%


\date{28 October 2026}

\begin{abstract}
For every centered one-mode continuous-variable Gaussian state, we prove that
the signed Wigner periodization defines a discrete-variable density operator
in every odd dimension $d\ge3$. Positivity is equivalent to the one-mode uncertainty
condition, and we determine the complete inertia outside and on its boundary.
A covariance-independent isometric embedding gives exponential convergence
in trace norm, with an explicit error bound uniform in Planck's constant.
For covariance twice the identity, dimension $31$ gives a certified trace-norm
error below $1/8000$.
The trace-norm estimate controls all measurement probabilities. We retain
the original canonical-coherent-frame formula in terms of the determinant
and trace of the covariance matrix, with continuous normalization fixed to
agree with the discrete frame convention. The coherence comparisons identify
which diagnostics accompany the controlled state approximation. Position- and
momentum-basis coherences have exponentially convergent values after division
by the square root of the dimension. Numerical comparisons of purity and
frame coherence illustrate the discrete-continuous transition.
\end{abstract}

\maketitle



\section{Introduction} The Gaussian quantum states play a fundamental role in quantum mechanics, quantum information, quantum computers and quantum optics \cite{Weedbrook12,Ferraro05,Wang07,Adesso14,Navarrete22}. 
An adequate description of their coherence is important in the case of certain applications \cite{Xu16,Tan17,Ahnefeld26}.
In a Hilbert space, an orthonormal basis is a set of orthogonal unit vectors satisfying the resolution of the identity.
The elements of the Hilbert space as well as the linear operators can be
easily described in terms of an orthonormal basis, and for each of them 
the description can be done in a single way. So, we have uniquely determined coordinates for each element or linear operator. 
Replacing orthonormality by a frame operator proportional to the identity
gives the more general concept of a tight frame
\cite{Christensen03,Waldron18,Ferrie08}. A tight frame satisfies
$\sum_j|\varphi_j\rangle\langle\varphi_j|=\kappa I$ for some $\kappa>0$.
Rescaling its vectors by $\kappa^{-1/2}$ gives a Parseval frame and the
associated resolution of the identity.
Any orthonormal basis is a tight frame, but a tight frame is a more general concept. For example, a tight frame can contain several orthogonal bases, identical or null vectors.\\
The representation of an element of the Hilbert space or of a linear operator in terms of a tight frame is not unique, but there exists 
a standard representation which can be easily obtained by using the 
resolution of the identity satisfied by the vectors of the frame.
This allows us to extend certain definitions from orthonormal bases to tight frames.\\
There exist several measures for the coherence of a quantum state 
\cite{Baumgratz14,Chen16,Hu17,Sperling18,Ma19,Radhakrishnan19,Bischof19,Mandal20,Cotfas25}  , but in this article, we  use only the extension of $\ell_1$-norm of coherence \cite{Baumgratz14,Chen16} to tight frames, we call {\it frame-dependent coherence}. 
Several examples concerning the finite-dimensional case are presented in \cite{Cotfas25}.
After normalization, a tight frame defines a rank-one POVM (positive operator-valued measure describing a generalized measurement). The frame-dependent functional is related to
the usual off-diagonal $\ell_1$ coherence and its POVM-based extension
\cite{Baumgratz14,Chen16,Bischof19,Cotfas25} by the explicit affine
normalization given in Sec.~\ref{sec-frame-functional}.  In this article,
we consider only single-mode, pure or mixed, continuous-variable Gaussian states (CVGS) \cite{Weedbrook12,Ferraro05,Wang07,Adesso14,Navarrete22} or
discrete-variable Gaussian states (DVGS) \cite{Cotfas23,Vourdas04}.\\
Based on a formula from \cite{Cotfas12}, we have obtained in  \cite{Cotfas23} the representation (\ref{Wgausskappa}) of the Wigner function of a pure DVGS ${\tilde{g}}_\kappa$ in terms of
the Wigner function of the corresponding CVGS ${g}_\kappa$. By extending (\ref{Wgausskappa}) from pure to mixed case,
we have associated a mixed DVGS  to any mixed CVGS  in any odd dimension $d=3,5,7,...$.
The earlier discrete-continuous correspondence in \cite{Cotfas23}
motivates the approximation studied here. Its Supplemental Material,
Theorems~5 and 6, gives the qualitative state and purity limits;
Supplemental Theorem~7 of its arXiv version states purity on the boundary
$D=1$. Our direct proof places that boundary statement within the complete
inertia classification.
We prove positivity for the exact
signed construction, specify a common Hilbert space through an explicit
isometric embedding, and establish exponential convergence in trace norm.
The embedding is independent of the covariance, while the error bound depends
on its position and momentum widths. Theorem~\ref{convergence} is the main
approximation result. The numerical data in Table~\ref{purityc} illustrate
purity and coherent-frame coherence; the basis comparisons are given in
Table~\ref{comparison}.

Finite Nyquist--Shannon representations of oscillator wavefunctions and
field/conjugate-field sampling error estimates were developed in
\cite{macridin18,macridin22}. They provide the sampling context for
Appendix~\ref{app-convergence}.
Continuous--discrete Wigner correspondences have also been developed through
Gottesman--Kitaev--Preskill encoding \cite{fengluo24}.
Exponential finite-rank approximation alone
also follows from the thermal spectrum: for $D=\det\sigma\ge1$ and
$r=(\sqrt D-1)/(\sqrt D+1)$, the normalized first-$d$ spectral truncation
has trace-norm error $2r^d$. That spectral subspace depends on $\sigma$.
Here the prescribed signed Wigner family has an exact physicality and
inertia theorem, and its reconstruction uses a map independent of $\sigma$
with an explicit mixed-state error bound.

\section{Continuous-variable Gaussian states}
\subsection{One-dimensional quantum systems} The mathematical description 
of these systems uses the Hilbert space
\begin{equation}
L^2(\mathbb{R})\!=\!\overline{\left\{ 
\Psi\!:\!\mathbb{R}\!\rightarrow\!\mathbb{C}\ \left| \ 
\int_{-\infty}^\infty |\Psi(q)|^2\, dq\!<\!\infty\right.\right\}}
\end{equation}
(closure of the space of square integrable functions) 
 with the inner product 
\begin{equation}
\left\langle \Phi,\Psi\right\rangle\!=\!
\int_{-\infty}^\infty \overline{\Phi(q)}\, \Psi(q)\, dq,
\end{equation}
Fourier transform \cite{Planck}
\begin{equation}
(F\Psi)(p)\!=\!\frac{1}{\sqrt{h}}\int\limits_{-\infty}^\infty \mathrm{e}^{-\frac{2\pi \mathrm{i}}{h}pq}\, \Psi(q)\, dq, 
\end{equation}
and displacement operators
\begin{equation} 
D(q,p)\Psi(x)=\mathrm{e}^{-\frac{\pi \mathrm{i}}{h}pq}\ 
\mathrm{e}^{\frac{2\pi \mathrm{i}}{h}px}\ \Psi(x-q).
\end{equation}
The states of the quantum system are described by density operators, that is by operators
$\varrho$ satisfying the relations: $\varrho=\varrho^\dag$ (Hermitian), $\varrho\geq 0$ (positive) and $\mathrm{tr}\, \varrho=1$ (unit trace).
The displaced parity operators 
\begin{equation} \label{Piqp}
\Pi(q,p)\Psi(x)=\mathrm{e}^{\frac{4\pi \mathrm{i}}{h}p(x-q)}\Psi(2q-x).
\end{equation}
provide a generalized Weyl--Wigner operator representation.
For the Gaussian density operators considered here, the Wigner function
\begin{equation} \label{cWigner}
\mathcal{W}_\varrho(q,p)=\frac{2}{h}\int\limits_{-\infty}^\infty \mathrm{e}^{-\frac{4\pi \mathrm{i}}{h}px}\,\langle q+x|\varrho|q-x\rangle \, dx
\end{equation}
satisfies the relation
\begin{equation} \label{rhosigma}
\varrho=2\int\limits_{\mathbb{R}^2}\mathcal{W}_{\!\varrho} (q,p)\,\Pi(q,p)\,dqdp.
\end{equation}

\subsection{Mixed continuous-variable Gaussian states} For any $\kappa \!\in\!(0,\infty)$, the normalized Gaussian function  
$g_\kappa :\mathbb{R}\rightarrow \mathbb{R}$, 
\begin{equation}
\begin{array}{l}
g_\kappa (q)\!=\! \sqrt[4]{\frac{2\kappa}{h}}\ \mathrm{e}^{-\frac{\kappa \pi}{h}q^2},
\end{array}
\end{equation}
is the wavefunction of a pure state with the density operator 
$\varrho_\kappa=|g_\kappa\rangle\langle g_\kappa|$
, and its Wigner function is
\begin{equation}\label{Fgkc}
\begin{array}{l}
\mathcal{W}_{\varrho_\kappa}(q,p)=\frac{2}{h}\, \mathrm{e}^{-\frac{2\pi}{h}(\kappa q^2+\frac{1}{\kappa} p^2)}.
\end{array}
\end{equation}
{\it Definition 1.} In the zero-displacement case, 
a CVGS  $\rho_\sigma$ is a state with a Wigner function of the form \cite{Weedbrook12,Ferraro05,Wang07,Adesso14,Navarrete22}
\begin{equation}
\!\!\!\begin{array}{l}
\mathcal{W}_{\!\varrho_\sigma} (q,p)\!=\!\frac{2}{h \sqrt{\det\sigma}}\, {\rm exp}\left\{\!-\frac{2\pi }{h} {\scriptstyle (\,q\ \ p\, )}{\scriptstyle \left(\!\!\!
\begin{array}{ll}
\sigma_{11} & \!\sigma_{12}\\
\sigma_{12} & \!\sigma_{22}
\end{array}\!\!\!\right)^{\!\!-1} \!\!\!\left(\!\!
\begin{array}{c}
q\\
p
\end{array}\!\!\right)}\!\right\},
\end{array}
\end{equation}
where  the covariance matrix
\begin{equation}\label{sigcon}
\sigma\!=\!{\scriptstyle \left(\!\!
\begin{array}{cc}
\sigma_{11} & \sigma_{12}\\
\sigma_{12} & \sigma_{22}
\end{array}\!\!\right)},\quad \mbox{satisfying}\quad \sigma \!+\!{\rm i} {\scriptstyle \left(\!\!
\begin{array}{cc}
0 & 1\\
-1 & 0
\end{array}\!\!\right)}\!\geq \!0,
\end{equation}
is a real $2\!\times\!2$ symmetric positive-definite matrix.

\section{Discrete-variable Gaussian states}
\subsection{Quantum systems in odd dimension} In order to have a good analogy with the continuous variable case,  we restrict ourselves to the case when the Hilbert space of the quantum system is an odd-dimensional ($d=2s+1$)  space $\mathcal{H}_d$, which we represent
as the space of all the functions 
$\psi:\{-s,-s+1,...,s-1,s\}\rightarrow \mathbb{C}$ with the scalar product 
\begin{equation} 
\langle\varphi|\psi\rangle=\sum\limits_{n=-s}^s\overline{\varphi(n)}\psi(n),
\end{equation}
Fourier transform
\begin{equation}
(\tilde{F}\psi)(k)\!=\!\frac{1}{\sqrt{d}}\sum\limits_{n=-s}^s \mathrm{e}^{-\frac{2\pi \mathrm{i}}{d}kn}\, \psi(n), 
\end{equation}
and discrete displacement operators \cite{Schwinger60,Atakishiev08}
\begin{equation}
\begin{array}{l}
 \tilde{D}(n ,k )\psi (m)=\mathrm{e}^{-\frac{\pi \mathrm{i}}{d}nk}\, \mathrm{e}^{\frac{2\pi \mathrm{i}}{d}km}\, \psi (m\!-\!n ).
\end{array}
\end{equation}
The discrete-variable quantum states are described by density operators, that is, by operators
$\tilde{\varrho} :\mathcal{H}_d\rightarrow\mathcal{H}_d$ satisfying the conditions
$\tilde{\varrho}^\dag=\tilde{\varrho}$, \, $\tilde{\varrho}\geq0$ \, and \, $\mathrm{tr}\,\tilde{\varrho}=1$.\\ 
{\it Definition 2.} The {\it discrete Wigner function} is defined as
\cite{Wootters87,Cohendet88,Opatrny95,Leonhardt95,Galetti96,Luis98,Tilma16}
\begin{equation}\label{dWigner}
\mathbf{\tilde{\mathcal{W}}}_{\tilde{\varrho}}(n,k)=\!\frac{1}{d}\,\sum\limits_{m=-s}^s {\rm e}^{-\frac{4\pi {\rm i}}{d}km}\, \langle n\!+\!m|\tilde{\varrho}|n\!-\!m\rangle,
\end{equation}
where $\{|-s\rangle,|-s+1\rangle,...,|s-1\rangle ,|s\rangle   \}$ is the canonical basis and the addition is modulo $d$.
It satisfies the relations \cite{Vourdas04}
\begin{eqnarray}
\mathbf{\tilde{\mathcal{W}}}_{\tilde{\varrho}}(n,k) & = & \frac{1}{d}\, \mathrm{tr}(\tilde{\varrho}\, \tilde{\Pi}(n,k)),\\
\tilde{\varrho} & = & \sum\limits_{n,k=-s}^s\mathbf{\tilde{\mathcal{W}}}_{\tilde{\varrho}}(n,k)\ \tilde{\Pi}(n,k), \label{varrhoPi}
\end{eqnarray}
where $\tilde{\Pi}(n,k)$ is the discrete displaced parity operator
\begin{equation} 
\tilde{\Pi}(n,k)\psi(m)=\mathrm{e}^{\frac{4\pi \mathrm{i}}{d}k(m-n)}
\psi(2n-m).
\end{equation}

\begin{figure}[t]
\includegraphics[width=\columnwidth]{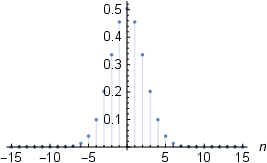}\\
\includegraphics[width=\columnwidth]{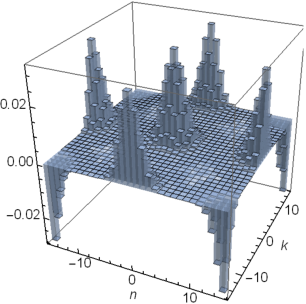}
\caption{The discrete-variable Gaussian function $\tilde{g}_1$ obtained by using (\ref{gausskappa}) and its Wigner function $\tilde{W}_{\!\tilde{\varrho}_1}$ obtained by using (\ref{Wgausskappa}), in the case $\kappa=1$ and $d\!=\!31$. }
\end{figure}

\subsection{Pure discrete-variable Gaussian states} For any $\kappa\in (0,\infty)$, the function \\ $\tilde{g}_\kappa:\{-s,-s+1,...,s-1,s\}\rightarrow \mathbb{C}$, 
\begin{equation} \label{gausskappa}
\tilde{g}_\kappa(n)=c_\kappa \sum\limits_{\ell=-\infty}^\infty
\mathrm{e}^{-\frac{\kappa \pi}{d}(n+\ell d)^2},
\end{equation}
where $c_\kappa$ is a normalizing constant (see Fig. 1), can be regarded \cite{Vourdas04} as a discrete version of $g_\kappa$. 
From a mathematical point of view, it is \cite{Vourdas04,Zhang04,Ruzzi06} a particular case of Jacobi function $\theta_3$. 
In \cite{Cotfas23}, we have proved that the discrete Wigner function $\tilde{\mathcal{W}}_{\!\tilde{\varrho}_\kappa}$ (see Fig. 1) of $\tilde{\varrho}_\kappa=|\tilde{g}_\kappa\rangle\langle\tilde{g}_\kappa|$  can be expressed in terms of the corresponding (see (\ref{Fgkc})) continuous Wigner 
function $\mathcal{W}_{\varrho_\kappa} $ of $g_\kappa$ as
(with $\eta=\sqrt{h/d}$)
\begin{equation}\label{Wgausskappa}
\begin{split}
 \tilde{\mathcal W}_{\tilde\varrho_\kappa}(n,k)
 &=c'_\kappa\sum_{\alpha,\beta\in\Z}(-1)^{\alpha\beta}\\
 &\quad\times\mathcal W_{\varrho_\kappa}
 \bigl((n+\alpha d/2)\eta,(k+\beta d/2)\eta\bigr).
\end{split}
\end{equation} 
where $c'_\kappa$ is a normalizing constant.

\subsection{Mixed discrete-variable Gaussian states} Inspired by the formulas (\ref{varrhoPi}) and (\ref{Wgausskappa}),
we have introduced in \cite{Cotfas23}, for any $d=3,5,7,...$, the discrete version
\begin{equation} \label{DVrhosigma}
\tilde{\varrho}_\sigma =\sum\limits_{n,k=-s}^s\tilde{\mathcal{W}}_{\tilde{\varrho}_\sigma}(n,k) \  \tilde{\Pi}(n,k),
\end{equation}
of CVGS $\varrho_\sigma$, where similar to (\ref{Wgausskappa}), we choose
\begin{equation}
\begin{split}
 \tilde{\mathcal W}_{\tilde\varrho_\sigma}(n,k)
 &=c'_\sigma\sum_{\alpha,\beta\in\Z}(-1)^{\alpha\beta}\\
 &\quad\times\mathcal W_{\varrho_\sigma}
 \bigl((n+\alpha d/2)\eta,(k+\beta d/2)\eta\bigr).
\end{split}
\end{equation}

\subsection{Kernel representation and positivity}
Write
\begin{equation}\label{covariance-entries}
 \sigma=\begin{pmatrix}a&c\\c&b\end{pmatrix}>0,
 \quad D=ab-c^2>0,\quad I_d=\{-s,\ldots,s\}.
\end{equation}
Here $a,b,c$ are the entries of $\sigma$. For $0<D<1$, we use the
same Gaussian expression only as a formal phase-space symbol and its
associated Hermitian kernel, without calling it a state. The inertia
statement below concerns $A_h$; the density-operator interpretation is used
only for $D\ge1$. Fourier inversion of the continuous
Wigner function gives the position kernel
\begin{equation}\label{kernel}
\begin{split}
 K_{\sigma,h}(x,y)&=\sqrt{\frac2{ha}}\exp\!\left[
 -\frac\pi{2ha}\{(x+y)^2+D(x-y)^2\}\right.\\
 &\hspace{32mm}\left.+\frac{i\pi c}{ha}(x^2-y^2)\right].
\end{split}
\end{equation}
Its diagonal integrates to one. The physical covariance is $h\sigma/(4\pi)$.
For the signed construction above, put $\eta=\sqrt{h/d}$ and define
\begin{align}
 X_h(n,k)&=\sum_{\alpha,\beta\in\mathbb Z}(-1)^{\alpha\beta}
 \mathcal W_{\varrho_\sigma}\bigl(\eta(n+\alpha d/2),
                                  \eta(k+\beta d/2)\bigr),\nonumber\\
 A_h&=\sum_{n,k\in I_d}X_h(n,k)\tilde\Pi(n,k),\label{source}\\
 U_{jk}&=\eta\sum_{r,t\in\mathbb Z}
 K_{\sigma,h}\bigl(\eta(j+dr),\eta(k+dt)\bigr).\label{periodized-kernel}
\end{align}
Poisson summation gives the exact identities
\begin{equation}\label{bridge}
 A_h=\frac{2d}{h}U,\qquad
 \tilde\varrho_\sigma=\frac{U}{\operatorname{tr}U}.
\end{equation}
The second identity applies whenever the trace is nonzero; the theorem below
proves it is strictly positive in the physical regime $D\ge1$.
Both $U$ and the normalized matrix are independent of $h$.
\begin{theorem}\label{positivity}
For every odd $d\ge3$ and all parameters above, $A_h\succeq0$ if and only if $D\ge1$, equivalently
$\sigma+i\left(\begin{smallmatrix}0&1\\-1&0\end{smallmatrix}\right)\succeq0$.
The numbers of positive, negative and zero eigenvalues are respectively
\[
\begin{array}{c|ccc}
 &+&-&0\\\hline
0<D<1&(d+1)/2&(d-1)/2&0\\
D=1&1&0&d-1\\
D>1&d&0&0.
\end{array}
\]
For $D<1$, $(e_1-e_{-1})/\sqrt2$ is a negative direction for every odd $d\ge3$.
\end{theorem}
The proof, including the exact Poisson identity, is given in
Appendix~\ref{app-positivity}. The sufficient condition suggested immediately
after Eq.~(22) of the arXiv version of \cite{Cotfas23} is thereby proved for
every odd dimension $d\ge3$; the theorem also establishes necessity and the
complete inertia. On the boundary $D=1$ the discrete
state is pure, as stated in Supplemental Theorem~7 of the arXiv version
of \cite{Cotfas23}; Appendix~\ref{app-positivity} gives a direct proof.
For $D>1$ it is positive definite.

\section{Quantitative approximation}
Assume henceforth $D\ge1$. Write $L=\sqrt{hd}$, $E_d=[-L/2,L/2)$, and define the isometry
\begin{equation}\label{embedding}
 (J_{d,h}e_j)(x)=\frac{\mathbf1_{E_d}(x)}{\sqrt{dL}}
 \sum_{k=-s}^s e^{2\pi ik(x/L-j/d)}.
\end{equation}
Its columns are the discrete Fourier transform of the orthonormal Fourier modes on $E_d$; their values are real by pairing $k$ and $-k$.
For $v>0$ set
\begin{align}
 B_d(v)&=\frac{2(2d/v)^{1/4}e^{-\pi d/(4v)}}{1-e^{-2\pi d/v}},\nonumber\\
 r_d(v)^2&=\frac{2\sqrt{2/(dv)}e^{-\pi(d+1)^2/(2dv)}}{1-e^{-2\pi(d+2)/(dv)}},\nonumber\\
 \epsilon_d&=e^{-\pi d/(4a)}+B_d(a)+\sqrt{B_d(b)^2+r_d(b)^2}.\label{epsilon}
\end{align}
\begin{theorem}\label{convergence}
Let $\varrho_\sigma$ have kernel \eqref{kernel} and let $F(\rho,\tau)=(\tr\sqrt{\sqrt\rho\tau\sqrt\rho})^2$ be squared fidelity. Then
\begin{align}
 \norm{J_{d,h}{\tilde{\varrho}_{\sigma}}J_{d,h}^*-\varrho_\sigma}_1&\le2\min(1,\epsilon_d),\label{tracebound}\\
 F(J_{d,h}{\tilde{\varrho}_{\sigma}}J_{d,h}^*,\varrho_\sigma)&\ge\max(0,1-\epsilon_d^2).
\end{align}
For fixed $\sigma$, $\epsilon_d=O_{a,b}(d^{1/4}e^{-\pi d/(4\max(a,b))})$, uniformly in $h$.
\end{theorem}
For $\sigma=2I$ and $d=31$, rational interval arithmetic certifies
$\epsilon_{31}<53952263/10^{12}<1/16000$.
Thus the trace-norm error is less than $1/8000$ and squared fidelity exceeds $1-1/256000000$, for every $h>0$.

Uniformity in $h$ is at fixed dimensionless covariance $\sigma$:
the dilation $(T_hf)(x)=h^{-1/4}f(x/\sqrt h)$ transforms both the state and
the embedding from $h=1$. The dimensional covariance $h\sigma/(4\pi)$
therefore varies with $h$ in this assertion.
For $R_d=J_{d,h}\tilde\varrho_\sigma J_{d,h}^*$, the operational consequence is
\begin{equation}\label{measurement-error}
 \sup_{0\le E\le I}|\tr E(R_d-\varrho_\sigma)|
 =\tfrac12\norm{R_d-\varrho_\sigma}_1
 \le\min(1,\epsilon_d).
\end{equation}
The proof is in Appendix~\ref{app-convergence}. It controls a common
purification by Gaussian position tails, translated position blocks,
Fourier aliases and omitted Fourier coefficients. In particular,
\begin{equation}\label{purity-error}
 \left|\tr\tilde\varrho_\sigma^2-D^{-1/2}\right|
 \le 4\min(1,\epsilon_d).
\end{equation}
Indeed, for density operators $R,S$, $|\tr(R^2-S^2)|\le2\norm{R-S}_1$.
The trace-norm statement alone does not give convergence of an entrywise
$\ell_1$ functional over a changing coherent frame. That comparison is
reported numerically below, separately from the proved state and purity limits.

\section{Frames and coherence}
\subsection{Tight frames} In a $d$-dimensional complex Hilbert space $\mathcal{H}$, an  orthonormal basis is a set of vectors
$\mathcal{B}=\{ \, |\varphi _1\rangle,\, |\varphi _2\rangle,\, ...\, |\varphi _d\rangle\, \}$
satisfying the relation 
$\langle \varphi _j|\varphi_k\rangle =\delta_{jk}$ and the resolution of the identity $ 
    \sum\limits_{j=1}^d|\varphi_j\rangle\langle\varphi_j|=\mathbb{I}_{\mathcal{H}}$, 
where $\mathbb{I}_{\mathcal{H}}:\mathcal{H}\rightarrow\mathcal{H},\ 
\mathbb{I}_{\mathcal{H}}|\psi\rangle =|\psi\rangle$ is the identity operator.
In this case, any $|\psi\rangle \in \mathcal{H}$ admits the representation 
\begin{equation}
    |\psi\rangle \equiv \mathbb{I}_{\mathcal{H}}|\psi\rangle =    \sum\limits_{j=1}^d|\varphi_j\rangle\langle\varphi_j|\psi\rangle,
\end{equation}
where $\langle\varphi_j|\psi\rangle$ are the coordinates of $|\psi\rangle$ in $\mathcal{B}$,
and any linear  operator $A:\mathcal{H}\rightarrow\mathcal{H}$
admits the representation 
\begin{equation}
    A \equiv \mathbb{I}_{\mathcal{H}}A\mathbb{I}_{\mathcal{H}} =    \sum\limits_{j,k=1}^d|\varphi_j\rangle\langle\varphi_j|A|\varphi_k\rangle\langle\varphi_k|,
\end{equation}
where $\langle\varphi_j|A|\varphi_k\rangle$ are the elements of the matrix of $A$.\\
A tight frame \cite{Christensen03,Waldron18,Ferrie08} in $\mathcal{H}$ is a set of vectors
$\mathcal{F}=\{ \, |\varphi _1\rangle,\, |\varphi _2\rangle,\, ...\, |\varphi _m\rangle\, \}$ satisfying the single condition
\begin{equation}\label{kappaI}
      \frac{1}{\kappa} \sum\limits_{j=1}^m|\varphi_j\rangle\langle\varphi_j|=\mathbb{I}_{\mathcal{H}}, 
\end{equation}
where $\kappa \in (0,\infty)$ is a constant.
Evidently, any orthonormal basis is a tight frame, but the converse is not true.\\
In this article, we consider only tight frames containing only unit vectors (vectors of norm 1). In this case, the relation (\ref{kappaI}) can be written as
\begin{equation}\label{dmI}
\begin{array}{l}
       \frac{d}{m}\sum\limits_{j=1}^m|\varphi_j\rangle\langle\varphi_j|= \mathbb{I}_{\mathcal{H}} 
\end{array}       
\end{equation}
(both sides of (\ref{kappaI}) must have the same trace).\\
Generally, the representation of a vector or operator in terms of a frame is not unique,
but there exist the privileged representations defined by the relations 
\begin{equation} \begin{array}{l}
    |\psi\rangle \equiv \mathbb{I}_{\mathcal{H}}|\psi\rangle =    \frac{d}{m}\sum\limits_{j=1}^m|\varphi_j\rangle\langle\varphi_j|\psi\rangle,\\
    A \equiv \mathbb{I}_{\mathcal{H}}A\mathbb{I}_{\mathcal{H}} =    \frac{d^2}{m^2}\sum\limits_{j,k=1}^m|\varphi_j\rangle\langle\varphi_j|A|\varphi_k\rangle\langle\varphi_k|.
    \end{array}
\end{equation}
{\it Example 1.} 
If  $\mathfrak{B}\!=\!\{ |\psi_1\rangle, \, |\psi_2\rangle,\, ...\, , \, |\psi_d\rangle\}$ and 
$\mathfrak{B}'\!=\!\{ |\varphi_1\rangle, \, |\varphi_2\rangle,\, ...\, , \, |\varphi_d\rangle\}$ are two 
orthonormal bases in $\mathcal{H}$, then 
$\mathfrak{F}\!=\!\{ |\psi_1\rangle, \, |\psi_2\rangle,\, ...\, , \, |\psi_d\rangle, |\varphi_1\rangle, \, |\varphi_2\rangle,\, ...\, , \, |\varphi_d\rangle\}$
is a tight frame in $\mathcal{H}$.\\
{\it Example 2.} 
The discrete coherent states \cite{Vourdas04,Tolar97,Klimov09} \\
$|n,k\rangle=\tilde{D}(n,k)\tilde{g}_1$, defined explicitly as
\begin{equation} 
\langle m|n,k\rangle=\mathrm{e}^{-\frac{\pi \mathrm{i}}{d}kn}\ 
\mathrm{e}^{\frac{2\pi \mathrm{i}}{d}km}\ \tilde{g}_1(m-n),
\end{equation}
form a tight frame in $\mathcal{H}_d$, namely
$\frac{1}{d}\!\!\sum\limits_{n,k=-s}^s\!\!|n,k\rangle\langle n,k|\!=\!\mathbb{I}.$\\
{\it Example 3.} 
The canonical coherent states 
$|q,p\rangle=D(q,p)|g_1\rangle$, where $g_1(x)=\sqrt[4]{2/h}\ \mathrm{e}^{-\frac{\pi}{h}x^2}$ is the groundstate of the quantum harmonic oscillator, 
is the most remarkable known tight frame \cite{Gazeau09} in the infinite-dimensional Hilbert space $L^2(\mathbb{R})$. The functions $|q,p\rangle$, explicitly described as
\begin{equation} \label{coherentst}
\langle x|q,p\rangle=\sqrt[4]{\frac{2}{h}}\ \mathrm{e}^{-\frac{\pi \mathrm{i}}{h}pq}\,
\mathrm{e}^{\frac{2\pi \mathrm{i}}{h}px}\ \mathrm{e}^{-\frac{\pi}{h}(x-q)^2},
\end{equation}
satisfy the identity 
$\frac{1}{h}\int\limits_{\mathbb{R}^2}|q,p\rangle\langle q,p| \, dqdp=\mathbb{I}.$
The usual factor $1/(2\pi)$ is recovered when $h=2\pi$.

\subsection{Basis-dependent coherence} 
The $\ell_1$-norm of coherence of a state $\varrho$ with respect to an orthonormal basis 
$\mathcal{B}=\{ \, |\psi _1\rangle,\, |\psi _2\rangle,\, ...\, |\psi _d\rangle\, \}$ is usually defined as \cite{Baumgratz14,Chen16}
\begin{equation} \label{basiscoherence}    
\mathcal{C}_{\ell_1,\mathcal B}(\varrho)=\sum\limits_{j\neq k}|\langle\psi_j|\varrho|\psi_k\rangle|.
\end{equation}
Since $\langle\psi_j|\varrho|\psi_j\rangle\geq 0$ and 
$\sum_{j=1}^d|\langle\psi_j|\varrho|\psi_j\rangle|=
\sum_{j=1}^d\langle\psi_j|\varrho|\psi_j\rangle=\mathrm{tr}\, \varrho =1$,
we have $\sum_{j\neq k}|\langle\psi_j|\varrho|\psi_k\rangle|=\sum_{j,k=1}^d|\langle\psi_j|\varrho|\psi_k\rangle|-1.$ Consequently,
the whole information concerning the $\ell_1$-norm of coherence of $\varrho$ is contained
in $\sum_{j,k=1}^d|\langle\varphi_j|\varrho|\varphi_k\rangle|$.
In order to have a good correspondence discrete $\longleftrightarrow$ continuous, in this article,
we define the $\ell_1$-based basis-dependent coherence functional of $\varrho$ as
\begin{equation} \label{basiscoherence1}    
\mathcal{C}_\mathcal{B}(\varrho)=\frac{1}{2}\sum\limits_{j,k=1}^d|\langle\psi_j|\varrho|\psi_k\rangle|.
\end{equation}

Thus our convention is the affine rescaling
$\mathcal C_{\mathcal B}=(1+\mathcal C_{\ell_1,\mathcal B})/2$ of the usual off-diagonal measure.
It equals $1/2$ on states diagonal in $\mathcal B$.
We use the term frame-dependent coherence for this specified entrywise kernel
functional; no additional faithfulness or resource-monotonicity assertion is
needed for the comparisons below.
Evidently, the value of the coherence of a state $\varrho$ depends on the basis 
$\mathcal{B}$ we choose.

\subsection{Frame-dependent coherence}\label{sec-frame-functional} Let $\mathcal{F}=\{ \, |\varphi _1\rangle,\, |\varphi _2\rangle,\, ...\, |\varphi _m\rangle\, \}$ be a fixed tight frame in the $d$-dimensional Hilbert space $\mathcal{H}$. 
Since 
$\sum\limits_{j=1}^m|\langle\varphi_j|\varrho|\varphi_j\rangle|=
\sum\limits_{j=1}^m\langle\varphi_j|\varrho|\varphi_j\rangle
=\sum\limits_{j=1}^m\mathrm{tr}(\varrho\, |\varphi_j\rangle\langle\varphi_j|)
=\mathrm{tr}(\varrho\,\sum\limits_{j=1}^m |\varphi_j\rangle\langle\varphi_j|)
\!=\!\frac{m}{d}\mathrm{tr}(\varrho \, \mathbb{I}) \!=\!\frac{m}{d}$, again $\sum_{j\neq k}|\langle\varphi_j|\varrho|\varphi_k\rangle|$ and 
$\sum_{j,k=1}^m|\langle\varphi_j|\varrho|\varphi_k\rangle|$ contain the same information.
The definition we propose for the frame-dependent coherence is a direct
extension of the basis-dependent $\ell_1$-based functional (\ref{basiscoherence1}).\\
{\it Definition 3.} The {\it $\ell_1$-based frame-dependent coherence functional} of a state $\varrho$ with respect to the frame $\mathcal{F}=\{ \, |\varphi _1\rangle,\, |\varphi _2\rangle,\, ...\, |\varphi _m\rangle\, \}$ is
\begin{equation} \label{framecoherence}    
\mathcal{C}_\mathcal{F}(\varrho)=\frac{d^2}{2m^2}\sum\limits_{j,k=1}^m|\langle\varphi_j|\varrho|\varphi_k\rangle|.
\end{equation}
To compare conventions, put
$\psi_j=\sqrt{d/m}\,\varphi_j$ and let
$(V\chi)_j=\langle\psi_j|\chi\rangle$ be the isometric analysis map.
The off-diagonal $\ell_1$ coherence of its Naimark state is
\begin{equation}\label{naimark-coherence}
 C_N(\varrho)=\frac dm\sum_{j,k=1}^m
 |\langle\varphi_j|\varrho|\varphi_k\rangle|-1.
\end{equation}
The diagonal sum is one in the Parseval convention, and therefore
\begin{equation}\label{frame-affine}
 \mathcal C_{\mathcal F}(\varrho)=\frac d{2m}(1+C_N(\varrho))
 =\tfrac12 C_{\rm prev}(\varrho)+\frac d{2m}.
\end{equation}
Here $C_{\rm prev}=(d/m)C_N$ is the functional of
\cite{Cotfas25}, Eqs.~(25) and (83), evaluated on the Parseval frame
$\{\psi_j\}$. The present convention retains the diagonal mass and its
corresponding baseline.
If the frame $\mathcal{F}$ is an orthonormal basis, then the frame-dependent coherence coincides with the basis-dependent coherence (\ref{basiscoherence1}).\\

\subsection{Frame-dependent coherence of a discrete-variable Gaussian state} 
The frame-dependent coherence of DVGS $\tilde{\varrho}_\sigma$ with respect to the frame $\tilde{\mathcal{F}}$ of all the discrete coherent states is 
\begin{equation} \label{dframcoher}
\tilde{\mathcal{C}}_{\tilde{\mathcal{F}} }(\tilde{\varrho}_\sigma )
=\frac{1}{2d^2}\sum\limits_{n,k,m,\ell=-s}^s|\langle n,k|\tilde{\varrho}_\sigma|m,\ell\rangle|,
\end{equation}
where
\begin{equation}
\begin{array}{l}
\langle n,k|\tilde{\varrho}_\sigma|m,\ell\rangle =\sum\limits_{a,b=-s}^s\langle n,k|a\rangle\langle a|\tilde{\varrho}_\sigma|b\rangle \langle b|m,\ell\rangle\\
\langle a|\tilde{\varrho}_\sigma|b\rangle =\sum\limits_{u,v=-s}^s\mathbf{\tilde{\mathcal{W}}}_{\tilde{\varrho}_\sigma}(u,v)\ \langle a|\tilde{\Pi}(u,v)|b\rangle,\\
\langle a|\tilde{\Pi}(u,v)|b\rangle=\tilde{\Pi}(u,v)\delta_b(a)=
\mathrm{e}^{\frac{2\pi \mathrm{i}}{d}v(a-b)}\delta_{2u}(a+b).
\end{array}
\end{equation}

\subsection{Frame-dependent coherence of a continuous-variable Gaussian state} By following the analogy with (\ref{dframcoher}), we define the frame-dependent coherence $\mathcal{C}_\mathcal{F}(\varrho_\sigma)$ of a CVGS $\varrho_\sigma$ with respect to the frame $\mathcal{F}$ of all the canonical coherent states as
\begin{equation} 
\mathcal{C}_\mathcal{F}(\varrho_\sigma)=\frac{1}{2h^2}\int\limits_{\mathbb{R}^4}|\langle \alpha,\beta|\varrho_\sigma|\gamma,\delta\rangle|\ d\alpha d\beta d\gamma d\delta.
\end{equation}
The measure of the canonical coherent frame is $dq\,dp/h$.
We retain the factor $1/2$ in the finite-frame convention
(\ref{framecoherence}), so that the continuous prefactor is $1/(2h^2)$.
This convention gives $\mathcal C_{\mathcal F}(|g_1\rangle\langle g_1|)=2$:
$|\langle q,p|g_1\rangle|=e^{-\pi(q^2+p^2)/(2h)}$ has integral $2h$.
It also retains the closed formula below and the discrete normalization
(\ref{dframcoher}). The Wigner inversion (\ref{rhosigma}) has its own factor
$2$, arising from the substitution $y=2q-x$; these are different changes of
variables.\\
The determinant-and-trace formula was stated in \cite{Cotfas26},
Theorem~1. Here the functional has the continuous normalization $1/(2h^2)$
specified above.
\begin{theorem}\label{continuousframe}
The frame-dependent coherence of CVGS $\varrho_\sigma$ with respect to the frame of canonical coherent states is
\begin{equation} \label{coherencecsigma}
\mathcal{C}_\mathcal{F}(\varrho_\sigma)=\frac{1}{\sqrt{\det\sigma}}\sqrt{\det \sigma + \mathrm{tr}\, \sigma +1},
\end{equation}
where \ $\mathrm{det}\, \sigma=\sigma_{11}\sigma_{22}-\sigma^2_{12}$ \ and \ 
$\mathrm{tr}\, \sigma =\sigma_{11}+\sigma_{22}$.
\end{theorem}
The detailed Gaussian calculation is in Appendix~\ref{app-frame}.

Table~\ref{purityc} presents the purity  $\mathrm{tr}\,{\varrho}^2_\sigma =\frac{1}{\sqrt{\mathrm{det}\, \sigma}}$ and the coherence $\mathcal{C}_{\mathcal{F}}(\varrho_\sigma )$ computed by using (\ref{coherencecsigma}) of a CVGS ${\varrho}_\sigma$ versus the values of purity $\mathrm{tr}\,\mathbf{\tilde{\varrho}}^2_\sigma  $ and coherence 
$\tilde{\mathcal{C}}_{\tilde{\mathcal{F}}}(\mathbf{\tilde{\varrho}}_\sigma )$ of the corresponding DVGS 
$\mathbf{\tilde{\varrho}}_\sigma$ for $d\in\{ 3,5,7,9,11,13,15\}$, in five particular cases. Purity convergence follows from
Eq.~(\ref{purity-error}). The frame values approach the continuous formula in
these examples; an all-covariance rate for this changing-frame functional is
a separate problem.
The values of $\mathrm{tr}\,\mathbf{\tilde{\varrho}}^2_\sigma  $ have been obtained by using the formula
\begin{equation}
\mathrm{tr}\,\mathbf{\tilde{\varrho}}^2_\sigma  =\sum\limits_{n,k=-s}^s\left|\langle n|\tilde{\varrho}_\sigma|k\rangle\right|^2,
\end{equation}
and those of $\tilde{\mathcal{C}}_{\tilde{\mathcal{F}} }(\tilde{\varrho}_\sigma )$ by using (\ref{dframcoher}).

\begin{table*}[t]
\caption{\label{purityc}Purity and frame coherence of the continuous Gaussian
state and its discrete approximations. The original five covariances and
dimensions are retained. Discrete frame entries are numerical evaluations
cross-checked through the signed Wigner and periodized-kernel representations;
they illustrate convergence, while Theorem~\ref{convergence} proves trace-norm
convergence and Eq.~(\ref{purity-error}) controls purity.}
\begin{ruledtabular}
\begin{tabular}{cccccc}
$\sigma$ & $\tr\varrho_\sigma^2$ & $\mathcal C_{\mathcal F}(\varrho_\sigma)$
& $d$ & $\tr\tilde\varrho_\sigma^2$ & $\tilde{\mathcal C}_{\tilde{\mathcal F}}(\tilde\varrho_\sigma)$\\
\colrule
$\left(\begin{smallmatrix}7&-\pi\\-\pi&5\end{smallmatrix}\right)$ & 0.1995 & 1.2318 & $\begin{array}{c}3\\5\\7\\9\\11\\13\\15\end{array}$ & $\begin{array}{c}0.3390\\0.2333\\0.2081\\0.2017\\0.2001\\0.1996\\0.1995\end{array}$ & $\begin{array}{c}0.7497\\0.9616\\1.0778\\1.1475\\1.1867\\1.2077\\1.2188\end{array}$\\
\colrule
$2I$ & 0.5000 & 1.5000 & $\begin{array}{c}3\\5\\7\\9\\11\\13\\15\end{array}$ & $\begin{array}{c}0.5287\\0.5010\\0.5000\\0.5000\\0.5000\\0.5000\\0.5000\end{array}$ & $\begin{array}{c}0.8648\\1.2438\\1.3909\\1.4551\\1.4811\\1.4920\\1.4966\end{array}$\\
\colrule
$\left(\begin{smallmatrix}1&\sqrt3\\\sqrt3&6\end{smallmatrix}\right)$ & 0.5774 & 1.9149 & $\begin{array}{c}3\\5\\7\\9\\11\\13\\15\end{array}$ & $\begin{array}{c}0.7021\\0.5949\\0.5796\\0.5776\\0.5774\\0.5774\\0.5774\end{array}$ & $\begin{array}{c}1.0121\\1.3780\\1.5978\\1.7269\\1.8024\\1.8465\\1.8730\end{array}$\\
\colrule
$\left(\begin{smallmatrix}3&2\\2&2\end{smallmatrix}\right)$ & 0.7071 & 2.0000 & $\begin{array}{c}3\\5\\7\\9\\11\\13\\15\end{array}$ & $\begin{array}{c}0.6633\\0.7009\\0.7064\\0.7070\\0.7071\\0.7071\\0.7071\end{array}$ & $\begin{array}{c}1.0308\\1.5586\\1.8052\\1.9137\\1.9613\\1.9826\\1.9921\end{array}$\\
\colrule
$\left(\begin{smallmatrix}2&0\\0&1/2\end{smallmatrix}\right)$ & 1.0000 & 2.1213 & $\begin{array}{c}3\\5\\7\\9\\11\\13\\15\end{array}$ & $\begin{array}{c}1.0000\\1.0000\\1.0000\\1.0000\\1.0000\\1.0000\\1.0000\end{array}$ & $\begin{array}{c}1.1623\\1.6437\\1.8907\\2.0027\\2.0580\\2.0868\\2.1022\end{array}$\\
\end{tabular}\end{ruledtabular}\end{table*}

\section{Position and momentum coherence}
The basis and frame
quantities below examine diagnostics accompanying the state approximation.
Appendix~\ref{app-theta} also records exact theta representations and
imaginarity formulas, distinguishing finite-dimensional chirp aliasing from
the fixed-covariance limit.
The basis convention in Eq.~(\ref{basiscoherence1}) includes the diagonal:
\begin{equation}\label{coherence-convention}
 C_B(\rho)=\tfrac12\sum_{j,k}|\rho_{jk}|=\tfrac12(1+C_{\ell_1}(\rho)).
\end{equation}
In the continuous position representation define the kernel mass, with its scale explicit,
\begin{equation}\label{cv-coherence}
\begin{aligned}
 J_Q&=\iint|K_{\sigma,h}(x,y)|dxdy=\sqrt{\frac{2ha}{D}},\\
 \widehat C_Q&=\frac{J_Q}{2\sqrt h}=\sqrt{\frac a{2D}}.
\end{aligned}
\end{equation}
The change $u=(x+y)/2$, $v=x-y$ has unit Jacobian and proves this by two Gaussian integrals.
Fourier transformation sends $(a,b,c)$ to $(b,a,-c)$, so
$J_P=\sqrt{2hb/D}$ and $\widehat C_P=\sqrt{b/(2D)}$.
These finite kernel functionals require a reference scale; they are not ordinary coherence in a countable orthonormal basis.

There is no trace-one normal state exactly diagonal in position.
Such an operator would be multiplication by a nonzero $f$ on $L^2(\R)$.
In a region $|f|\ge\delta>0$, choose infinitely many disjoint positive-measure sets and their normalized indicators. Their images are orthogonal and have norms at least $\delta$, so multiplication by $f$ is not compact, hence not trace class.
Thus exact-position dephasing does not produce a normal diagonal density operator. The corresponding relative-entropy distance to normal diagonal states has an empty feasible set; it is $+\infty$ under the infimum convention. A finite-resolution channel gives a different, meaningful definition; see \cite{sajnok}.

\begin{theorem}\label{scaling}
Let $C_{Q,d}=\frac12\sum_{j,k\in I_d}|({\tilde{\varrho}_{\sigma}})_{jk}|$ and define $C_{P,d}$ in the discrete Fourier basis. For fixed $\sigma$,
\[
 \frac{C_{Q,d}}{\sqrt d}\longrightarrow\sqrt{\frac a{2D}},\qquad
 \frac{C_{P,d}}{\sqrt d}\longrightarrow\sqrt{\frac b{2D}}
\]
exponentially. In particular, the unscaled quantities diverge as $\sqrt d$.
Here is an explicit relative error bound. Set $\theta(t)=\sum_n e^{-\pi tn^2}$,
\begin{align*}
 Z_e&=\theta(ad/2)\theta(Dd/(2a))-1,&Q_e&=\theta(ad/D)^2-1,\\
 L_0&=\sqrt{2ad/D},&A_0&=1+\sqrt{ad},\\
 B_0&=\frac{2e^{-\pi(d+1)^2/(4ad)}}{1-e^{-\pi(d+2)/(ad)}},&
 E_0&=2\sqrt{2/(ad)}A_0B_0.
\end{align*}
Whenever $Z_e<1$,
\begin{equation}\label{scaling-bound}
 \left|\frac{C_{Q,d}}{\sqrt{ad/(2D)}}-1\right|
 \le\frac{Q_e+Z_e+2E_0/L_0}{1-Z_e}.
\end{equation}
The momentum bound replaces $a$ by $b$.
\end{theorem}
The proof is in Appendix~\ref{app-basis}.

Table~\ref{comparison} compares the continuous position and momentum kernel
masses with the rescaled discrete basis values. Its discrete entries have
outward-rounded Arb enclosures with an analytic Gaussian remainder, of width
at most $2\cdot10^{-11}$; the displayed values are rounded to eight decimals.
The accompanying data cover eight covariances at each
$d\in\{3,5,7,9,11,13,15,21,31,51\}$. The quantity $J_2$ is defined in
Appendix~\ref{app-theta}.

\section{Conclusions and outlook}
For the signed discrete-variable Gaussian construction we have proved the
one-mode uncertainty criterion for positivity in every odd dimension $d\ge3$ and
determined its full inertia. Our main approximation theorem embeds all these
finite matrices into $L^2(\mathbb R)$ by a covariance-independent isometry and
gives an explicit exponential trace-norm error bound. Purity convergence
follows, and position- and momentum-basis coherences converge exponentially
after division by $\sqrt d$.

We have presented a version for continuous-variable Gaussian states of the
frame-dependent coherence introduced in \cite{Cotfas25} for discrete-variable
quantum states. The explicit formula (\ref{coherencecsigma}) describes this
functional in terms of the covariance determinant and trace, with the
continuous and discrete prefactors fixed consistently. The original
coherent-frame comparisons and the new basis comparisons show the different
scalings involved. The numerical coherent-frame convergence remains distinct
from the trace-norm and basis-coherence theorems. Extensions to displaced and
multimode Gaussian states, and an explicit all-covariance error bound for the
coherent-frame functional, are natural further questions.

The transition from orthonormal bases to tight frames represents a natural
way to introduce redundant descriptions. The formula (\ref{rhosigma}) shows
that the Wigner representation uses displaced parity operators. By using
tight frames instead, a more general version of the Wigner function
\cite{DeBrota20,Ferrie11} can be defined \cite{Cotfas24}.


\section*{Data availability}
The original figure and covariance examples are from \cite{Cotfas26}.
The accompanying editable source archive \texttt{source\_20261028.zip}
contains the manuscript and original figures.
The verification package, available from the authors on request, contains
\texttt{verification/certificates.csv} (80 Arb-certified rows),
\texttt{verification/frame\_comparison.csv} (35 frame comparisons),
\texttt{verification/revision\_checks.json}, the programs, pinned requirements
and \texttt{README.md} with exact reproduction commands.
The certificate table uses 192-bit Arb arithmetic and analytic pre-normalization
Gaussian tails below $10^{-36}$. The frame table is independently
cross-checked double-precision computation. The finite numerical checks
support the analytic proofs and do not replace their all-parameter scope.

\section*{Author contributions}
Nicolae Cotfas provided the signed Wigner construction, the continuous-frame
coherence theorem and the original manuscript. Haobo Ma and Wenlin Zhang
contributed the positivity and quantitative approximation theorems, the
proofs in Appendices~A--D and the verification programs.

\section*{Declarations}
For Haobo Ma and Wenlin Zhang: This research received no specific funding.
The authors declare no competing interests.
Haobo Ma and Wenlin Zhang used AI tools, namely ChatGPT Pro (GPT-6 Pro), Codex (gpt-6.1-sol) and Claude Opus 5.5 (via Claude Code), for literature search, computation and double-checking. All proofs can be checked by hand, and the computations can be reproduced from the verification package. H. Ma and W. Zhang take full responsibility for their contributions.

\onecolumngrid
\begin{center}
\begin{minipage}{\textwidth}
\makeatletter\def\@captype{table}\makeatother
\centering\small
\caption{Purity, raw position coherence, rescaled basis coherences, and $J_{2,d}$.
The continuous rows list the rescaled targets; raw $C_Q$ diverges.}\label{comparison}
\begin{tabular}{llrrrrr}\toprule
$\sigma$&$d$&purity&$C_Q$&$C_Q/\sqrt d$&$C_P/\sqrt d$&$J_2$\\\midrule
$2I$&3&0.52865133&0.86591218&0.49993463&0.49993463&0.00000000\\
$2I$&15&0.50000000&1.93649167&0.50000000&0.50000000&0.00000000\\
$2I$&31&0.50000000&2.78388218&0.50000000&0.50000000&0.00000000\\
$2I$&CV&0.50000000&---&0.50000000&0.50000000&0.00000000\\\addlinespace
$\left(\begin{smallmatrix}3&2\\2&2\end{smallmatrix}\right)$&3&0.66328874&1.08270820&0.62510187&0.57183463&0.17563468\\
$\left(\begin{smallmatrix}3&2\\2&2\end{smallmatrix}\right)$&15&0.70710672&3.35004400&0.86497764&0.70703466&0.14943540\\
$\left(\begin{smallmatrix}3&2\\2&2\end{smallmatrix}\right)$&31&0.70710678&4.82181826&0.86602412&0.70710678&0.14942925\\
$\left(\begin{smallmatrix}3&2\\2&2\end{smallmatrix}\right)$&CV&0.70710678&---&0.86602540&0.70710678&0.14942925\\\addlinespace
$\left(\begin{smallmatrix}1&1\\1&2\end{smallmatrix}\right)$&3&1.00000000&1.14003562&0.65819987&0.81401716&0.22524480\\
$\left(\begin{smallmatrix}1&1\\1&2\end{smallmatrix}\right)$&15&1.00000000&2.73860400&0.70710451&0.99902166&0.14643896\\
$\left(\begin{smallmatrix}1&1\\1&2\end{smallmatrix}\right)$&31&1.00000000&3.93700394&0.70710678&0.99999868&0.14644661\\
$\left(\begin{smallmatrix}1&1\\1&2\end{smallmatrix}\right)$&CV&1.00000000&---&0.70710678&1.00000000&0.14644661\\\addlinespace
\bottomrule\end{tabular}
\end{minipage}
\end{center}
\appendix
\section{Positivity and exact signed periodization}\label{app-positivity}
We first identify the signed Wigner construction with its periodized kernel.
Scale to $h=d$, so $U_{ij}=\sum_{r,t}K_{\sigma,d}(i+dr,j+dt)$.
Choose the unique $n\in I_d$ and integer $\delta$ with $2n=i+j+d\delta$.
In \eqref{source}, set $u=\alpha+\delta$ and $v=2k+d\beta$.
Oddness of $d$ makes $(k,\beta)\mapsto v$ a bijection, and the sign and parity phase give
\[
 (A_d)_{ij}=\sum_{u,v\in\Z}
 \mathcal W_{\varrho_\sigma,d}\left(\frac{i+j+du}{2},\frac v2\right)
 e^{i\pi v((i-j)/d+u)}.
\]
For $q,w\in\R$, the Gaussian Fourier transform and Poisson summation give
\[
 \sum_{v\in\Z}\mathcal W_{\varrho_\sigma,d}(q,v/2)e^{i\pi vw}
 =2\sum_{m\in\Z}K_{\sigma,d}\bigl(q+d(w/2-m),q-d(w/2-m)\bigr).
\]
Substitute $q=(i+j+du)/2$, $w=(i-j)/d+u$, and then $r=u-m$, $t=m$.
This proves $A_d=2U$; scaling \eqref{cWigner} proves \eqref{bridge} for all $h$.
Positive definiteness of the real quadratic forms bounds all absolute summands by a Gaussian, justifying these rearrangements.

\begin{proof}
At $h=d$ put
$g(x)=(2/(da))^{1/4}e^{-\pi x^2/(da)+i\pi cx^2/(da)}$ and
$\lambda=\pi(D-1)/(2da)$. Then
$K_{\sigma,d}(x,y)=g(x)\overline{g(y)}e^{-\lambda(x-y)^2}$.
Here $[x]$ denotes the representative of $x\bmod d$ in $I_d$.
If $D>1$, for $v\in\C^d$,
\begin{equation}\label{positive-integral}
 v^*Uv=\frac1{\sqrt{4\pi\lambda}}
 \int_{\R}e^{-t^2/(4\lambda)}
 \left|\sum_{x\in\Z}\overline{v_{[x]}}g(x)e^{itx}\right|^2dt>0\quad(v\ne0).
\end{equation}
Strictness follows from uniqueness of the absolutely convergent Fourier series and $g(x)\ne0$.
For $D=1$, $U=zz^*$ with $z_j=\sum_r g(j+dr)$, and $z_0\ne0$.
To check the last point without a product formula, let
$\theta(z)=\sum_n e^{\pi izn^2}$ for $\operatorname{Im}z>0$.
Parity separation gives $\theta(z)\theta(z+1)=\theta(2z+1)^2$.
A zero would iterate to zeros with imaginary parts tending to infinity, contradicting
$|\theta(z)-1|\le2r/(1-r)<1$ when $r=e^{-\pi\operatorname{Im}z}<1/3$.
Apply this to $z=dc/a+id/a$.

Direct substitution gives
$K_{\sigma,d}(x,y)=D^{-1/2}K_{\sigma/D,d}(x,-y)$ and hence
$U_\sigma=D^{-1/2}U_{\sigma/D}R$, where $(Rv)_j=v_{-j}$.
The matrix $U_{\sigma/D}$ commutes with $R$ by simultaneous sign reversal.
If $D<1$, it is positive definite since $\det(\sigma/D)=1/D>1$.
Thus $U_\sigma$ is positive on the even subspace and negative on the odd subspace, of dimensions $s+1$ and $s$. Finally, the displayed uncertainty matrix has determinant $D-1$ and positive diagonal entries.
\end{proof}
Restricting \eqref{positive-integral} to $[-\pi,\pi]$ and applying Parseval gives
$v^*Uv\ge\sqrt{\pi/\lambda}e^{-\pi^2/(4\lambda)}\sum_x|v_{[x]}|^2|g(x)|^2$.
Apply this to $U_{\sigma/D}$ and the unit odd vector to obtain the quantitative certificate
\begin{equation}\label{negative}
 v^*U_\sigma v\le-\frac2{\sqrt{1-D}}
 \exp\!\left[-\frac{\pi da}{2(1-D)}-\frac{2\pi D}{da}\right]<0.
\end{equation}
For $d=3$, $\sigma=I/2$, and $u=e_1-e_{-1}$ this implies
$u^*Uu\le-8e^{-4\pi/3}/\sqrt3<-1/64$ and $u^*A_3u<-1/32$.
Indeed $\sqrt3<2$, $\pi<22/7$, $e<3$, and $4/243>1/64$ suffice.
For the thermal family $\sigma=\nu I$, the trace of $U$ is positive for every $\nu>0$ because all diagonal summands are positive. Its normalized matrix is therefore indefinite for $0<\nu<1$.

\section{Exponential approximation}\label{app-convergence}
\begin{proof}
Let $\mu$ be Gaussian of variance $\pi(D-1)/(ha)$, including a point mass when $D=1$.
Set $\psi_t(x)=(2/(ha))^{1/4}e^{-\pi x^2/(ha)+i\pi cx^2/(ha)+itx}$ and
$w_t(j)=\sqrt\eta\sum_r\psi_t(\eta(j+dr))$.
The characteristic function of $\mu$ and \eqref{bridge} give
$\varrho_\sigma=\int|\psi_t\rangle\langle\psi_t|d\mu$ and
$U=\int|w_t\rangle\langle w_t|d\mu$.
In $\mathcal H=L^2(\mu;L^2(\R))$, put $V(t)=\psi_t$, $Y(t)=J_{d,h}w_t$.
The exact marginals are
\[
 \mathbb E|\psi_t(x)|^2=\sqrt{\frac2{ha}}e^{-2\pi x^2/(ha)},\qquad
 \mathbb E|\widehat\psi_t(p)|^2=\sqrt{\frac2{hb}}e^{-2\pi p^2/(hb)},
\]
where $\widehat f(p)=h^{-1/2}\int f(x)e^{-2\pi ipx/h}dx$.
The second variance is the chirped-Gaussian variance $h(1+c^2)/(4\pi a)$ plus modulation variance $h(D-1)/(4\pi a)$, namely $hb/(4\pi)$.

Truncation to $E_d$ costs at most $e^{-\pi d/(4a)}$ in $\mathcal H$:
the squared tail is $\operatorname{erfc}(\sqrt{\pi d/(2a)})\le e^{-\pi d/(2a)}$.
Periodize $f_t(x)=\sum_r\psi_t(x+rL)$ on $E_d$.
For $r\ne0$ its translated block has norm at most
$(2d/a)^{1/4}e^{-\pi d(|r|-1/2)^2/a}$.
Minkowski and $(r-1/2)^2\ge1/4+2(r-1)$, $r\ge1$, give a cost $B_d(a)$.
The Fourier coefficients of $f_t$ are $c_k(t)=\sqrt\eta\,\widehat\psi_t(k\eta)$.
Poisson summation shows that $Y$ retains coefficients $\sum_m c_{k+md}$ for $|k|\le s$.
Their mean squared magnitudes satisfy
$\mathbb E|c_k|^2=\sqrt{2/(db)}e^{-2\pi k^2/(db)}$.
Each nonzero alias block has norm at most
$(2d/b)^{1/4}e^{-\pi d(|m|-1/2)^2/b}$; their sum is bounded by $B_d(b)$.
The omitted coefficients have total squared norm at most $r_d(b)^2$, since the first omitted positive index is $(d+1)/2$ and successive squares increase by at least $d+2$.
The retained aliases and omitted coefficients occupy orthogonal Fourier
subspaces, so their squared norms add. The normalized vector $Y/\norm Y$
has reduced state $J_{d,h}U J_{d,h}^*/\tr U$.
These estimates give
$\norm{Y-V}_{\mathcal H}\le\epsilon_d$.

Now $\norm V=1$, $\norm Y^2=\tr U>0$, and projection onto the complement of the line through $Y$ gives
$1-|\langle Y/\norm Y,V\rangle|^2\le\norm{Y-V}^2$.
Partial-trace contractivity and the pure-state trace-distance formula prove \eqref{tracebound}; the purification characterization of fidelity \cite{fuchs} gives the fidelity bound. The explicit tails give the rate.
\end{proof}
\section{Position and momentum basis limits}\label{app-basis}
\begin{proof}
Write $k(x,y)=K_{\sigma,d}(x,y)$ and $S=\sum_{x,y\in\Z}|k(x,y)|$.
Poisson summation and the eigenvalues $1/(da),D/(da)$ of its real quadratic matrix give
$L_0\le S\le L_0(1+Q_e)$.
The Gaussian majorant
\[ |k(x,y)|\le\sqrt{2/(da)}e^{-\pi(x^2+y^2)/(da)} \]
bounds the mass outside $I_d^2$ by $E_0$:
the one-dimensional full sum is at most $A_0$, and its tail from $(d+1)/2$ is at most $B_0$.
Let $M=\sum_{j,k}|U_{jk}|$. The triangle inequality, with central lifts separated from all other lifts, gives $S-2E_0\le M\le S$.
A second Poisson summation gives the exact normalizer
\[
 Z=\tr U=\sum_{m,k\in\Z}(-1)^{mk}
 e^{-\pi d(ak^2-2ckm+bm^2)/2}.
\]
Complete the square as $a(k-cm/a)^2+Dm^2/a$.
The shifted Gaussian sum in $k$ is at most the unshifted sum, by its Poisson expansion and $\cos\le1$.
Separating $(m,k)=(0,0)$ yields $|Z-1|\le Z_e$.
Divide $M$ by $2Z$ to obtain \eqref{scaling-bound}. Each error is a polynomial times a decaying exponential for fixed covariance. Fourier covariance, also valid discretely, proves the momentum assertion.
\end{proof}For completeness, the usual relative-entropy coherence of the finite matrices obeys
\begin{equation}\label{entropy-asymptotic}
 C_{{\rm rel},Q}({\tilde{\varrho}_{\sigma}})=\tfrac12\log d+\tfrac12\log(ea/2)
 -g\bigl((\sqrt D-1)/2\bigr)+o(1),
 \quad g(t)=(t+1)\log(t+1)-t\log t,
\end{equation}
with natural logarithms and $g(0)=0$. To justify the entropy limit, the normalized diagonal of $U/\tr U$ differs in total variation by an exponentially small quantity from the normalized truncated samples $e^{-2\pi j^2/(da)}$: the extra diagonal lifts are bounded by $E_0$ above.
Poisson summation, its derivative, and Gaussian second-moment tails give their entropy
$\frac12\log(da/2)+\frac12+o(1)$.
The finite-alphabet entropy continuity bound \cite{audenaert} makes the error tend to zero, since the total variation error times $\log d$ tends to zero.
For the state entropy, Gaussian symplectic diagonalization reduces the state to the thermal spectrum in \cite{Cotfas23}: it has geometric eigenvalues $(1-r)r^n$, $r=(\sqrt D-1)/(\sqrt D+1)$, so entropy $g((\sqrt D-1)/2)$.
Its normalized first-$d$ spectral truncation is within trace norm $2r^d$ (zero tail at $r=0$).
Together with \eqref{tracebound}, the two rank-at-most-$d$ states lie in a space of dimension at most $2d$ and have exponentially small trace distance.
The finite-dimensional entropy continuity bound \cite{audenaert} therefore makes their entropy difference $o(1)$; direct summation of the geometric spectrum gives the desired entropy limit.
Replacing $a$ by $b$ gives the momentum formula.

\section{Exact formulas and imaginarity}\label{app-theta}
\subsection{Theta formulas}
For a complex symmetric matrix with positive-definite real part write
$\vartheta_m(M)=\sum_{n\in\Z^m}e^{-\pi n^TMn}$.
Define
\begin{equation}\label{theta-data}
 H=\frac1{2da}\begin{pmatrix}1+D-2ic&1-D\\1-D&1+D+2ic\end{pmatrix},\quad
 G_{jk}=\sum_{r,t\in\Z}e^{-\pi(j+dr,k+dt)H(j+dr,k+dt)^T},\quad
 L_1=\begin{pmatrix}1&d\\1&0\end{pmatrix}.
\end{equation}
Then $U=\sqrt{2/(da)}G$, ${\tilde{\varrho}_{\sigma}}=G/T$, and
\begin{equation}\label{theta-trace}
 T=\tr G=\vartheta_2(L_1^THL_1)>0.
\end{equation}
Indeed the trace collects exactly the pairs $x\equiv y\pmod d$, uniquely $(x,y)=(n+dm,n)$.
Each entry is a shifted genus-two theta sum.

Put
\[
 M_- =\begin{pmatrix}H+\overline H&d\overline H\\d\overline H&d^2\overline H\end{pmatrix},\qquad
 M_+ =\begin{pmatrix}2H&dH\\dH&d^2H\end{pmatrix}.
\]
Their real parts are positive definite: their quadratic forms are the real parts of
$z^THz+(z+dt)^T\overline H(z+dt)$ and
$z^THz+(z+dt)^TH(z+dt)$, respectively.
Expanding two entries with the same residues parametrizes the lifts uniquely as $z,z+dt$, and proves
\begin{equation}\label{theta-measures}
 \tr{\tilde{\varrho}_{\sigma}}^2=\frac{\vartheta_4(M_-)}{T^2},\qquad
 J_{2,d}:=\tfrac14\norm{{\tilde{\varrho}_{\sigma}}-\overline{{\tilde{\varrho}_{\sigma}}}}_2^2
 =\frac{\vartheta_4(M_-)-\vartheta_4(M_+)}{2T^2}.
\end{equation}
Here $\vartheta_4(M_+)=\sum_{jk}G_{jk}^2=\tr(G\overline G)$.
$J_2$ is a Hilbert--Schmidt diagnostic; no resource monotonicity is claimed.
An independent purity expression from parity orthogonality is
\[\tr{\tilde{\varrho}_{\sigma}}^2=\frac{d\sum_{n,k}X_h(n,k)^2}{(\sum_{n,k}X_h(n,k))^2}.\]

The exact basis-coherence expressions are
\begin{equation}\label{theta-coherence}
 C_{Q,d}=\frac1{2T}\sum_{j,k\in I_d}|G_{jk}|,\qquad
 C_{Q,d}\big|_{c=0}=\frac{\vartheta_2(H)}{2\vartheta_2(L_1^THL_1)}.
\end{equation}
For $c=0$ all summands are positive, allowing the collapse to one theta quotient.
The finite Fourier transform $F_d$ obeys
$F_d{\tilde{\varrho}_{\sigma}}(a,b,c)F_d^*={\tilde{\varrho}_{\sigma}}(b,a,-c)$, as follows by rotating the discrete parity operators in \eqref{source}.
This also supplies $C_{P,d}$ explicitly.
For the discrete coherent frame, let $v_{nk}=\tilde D(n,k)v_{00}$ with
$(v_{00})_j$ proportional to $\sum_r e^{-\pi(j+dr)^2/d}$ and normalized to one.
The exact frame expression is
\[
 C_{{\rm frame},d}=\frac1{2d^2T}\sum_{n,k,m,\ell\in I_d}
 |v_{nk}^*Gv_{m\ell}|.
\]
It is a finite sum of theta moduli, rather than a single scalar theta value.

\subsection{Imaginarity}
Fix the antiunitary $(K\psi)(x)=\overline{\psi(x)}$ and write $\rho^*=K\rho K$.
It is also conjugation in the real Hermite basis. Equation \eqref{kernel} gives
$\varrho_{a,b,c}^*=\varrho_{a,b,-c}$ and $\rho=\rho^*\iff c=0$, in agreement with Xu \cite{xu}.
Momentum-basis conjugation is $RK$ in position space; since centered states commute with parity $R$, their position- and momentum-basis standard imaginarities coincide.
Let $f(\rho,\tau)=\tr\sqrt{\sqrt\rho\tau\sqrt\rho}$ and $F=f^2$.
The standard robustness, geometric and relative-entropy quantities are
\begin{equation}\label{imag-defs}
 I_R=\tfrac12\norm{\rho-\rho^*}_1,\qquad
 I_g=\tfrac12(1-f(\rho,\rho^*)),\qquad
 I_{\rm rel}=S((\rho+\rho^*)/2)-S(\rho).
\end{equation}
The finite-dimensional robustness identity is \cite{wu}, and the geometric identity is Theorem~3 of \cite{kondra}; we use squared fidelity $F$ consistently.
The robustness identity extends to trace class: if $A=\rho-\rho^*$, then
$KA_+K=A_-$ and adding $A_-$ to $\rho$ makes it real, with noise weight $\tr A_- =\norm A_1/2$; the triangle inequality gives the converse bound.
The geometric identity extends by real finite-rank truncations and fidelity continuity. Equivalently it is the continuous extension of the finite-dimensional geometric measure.
For relative entropy, the equality follows from
$\tr\rho\log\tau=\tr((\rho+\rho^*)/2)\log\tau$ for real $\tau$, and relative-entropy minimization.
All entropies here are finite for the Gaussian states and their two-state mixtures.

Specializing Xu's one-mode fidelity formula (Corollary~1 of \cite{xu}) gives
\begin{equation}\label{fidelity-gaussian}
 F_\sigma=\frac2{\sqrt{(D+1)^2+4c^2}-D+1},\qquad
 I_g=\frac{1-\sqrt{F_\sigma}}2,\qquad
 J_2=\frac12\left(\frac1{\sqrt D}-\frac1{\sqrt{D+c^2}}\right).
\end{equation}
Here $\sigma^K=S\sigma S$ with $S=\operatorname{diag}(1,-1)$, so that
$\sigma^K=\left(\begin{smallmatrix}a&-c\\-c&b\end{smallmatrix}\right)$ is the
covariance of the conjugate state; it is not entrywise conjugation of the real
matrix $\sigma$.
For the fidelity substitution,
$\det(\sigma+\sigma^K)=4ab$ and
$4ab+(D-1)^2=(D+1)^2+4c^2$.
For $J_2$, Gaussian integration gives $\tr\rho^2=1/\sqrt D$ and
$\tr(\rho\rho^*)=1/\sqrt{ab}$.
Thus these centered Gaussian quantities depend on $D$ and $|c|$, not $c$ alone at variable purity. Xu's fidelity measure $1-f$ is twice $I_g$.
For mixed states the useful rigorous bounds are
$1-\sqrt{F_\sigma}\le I_R\le\sqrt{1-F_\sigma}$ \cite{fuchs}.
For pure states, $D=1$, put $s_\infty=(1+c^2)^{-1/4}$. Then
\begin{equation}\label{pure-imag}
 I_R=\sqrt{1-s_\infty^2},\qquad I_g=(1-s_\infty)/2,\qquad
 I_{\rm rel}=h_2((1+s_\infty)/2),
\end{equation}
where $h_2(t)=-t\log t-(1-t)\log(1-t)$.
These follow from the overlap $|\langle\psi,\psi^*\rangle|=s_\infty$ and the two eigenvalues of the average of the two pure states.
For $c\ne0$, $(\rho+\rho^*)/2$ is generally not Gaussian: its symmetrically ordered $q^2p^2$ moment differs from the Gaussian with averaged covariance by $2(hc/(4\pi))^2$.
Its entropy therefore cannot be replaced by the entropy of that averaged covariance.

The discrete versions of \eqref{imag-defs} are exact finite-matrix expressions using \eqref{theta-data}.
At $D=1$, define
\[
 u_j=\sum_r e^{-\pi(1-ic)(j+dr)^2/(da)},\qquad
 s_d=\frac{|\sum_j u_j^2|}{\sum_j|u_j|^2}.
\]
Then $G=uu^*$ and all three formulas \eqref{pure-imag} hold with $s_\infty$ replaced by $s_d$; both sums are genus-two theta expressions.
An important finite-dimensional exception is exact:
with $a,D,d$ fixed, replacing $c$ by $c+2adk$, $k\in\Z$, and choosing $b=(D+c^2)/a$, leaves every lattice kernel summand unchanged.
Its phase multiplier is $e^{2\pi ik(x^2-y^2)}=1$ for integer $x,y$.
For example $d=3$ and $\sigma=\left(\begin{smallmatrix}1&6\\6&37\end{smallmatrix}\right)$ give exactly the vacuum discrete state, so discrete imaginarity is zero although $c=6$.
For fixed covariance this alias does not obstruct convergence:
\eqref{embedding} is real, hence $KJ_{d,h}=J_{d,h}K_d$ and
\[
 |I_R({\tilde{\varrho}_{\sigma}})-I_R(\varrho_\sigma)|\le 2\min(1,\epsilon_d),\qquad
 |\tr{\tilde{\varrho}_{\sigma}}^2-D^{-1/2}|\le4\min(1,\epsilon_d).
\]
Both follow from trace-norm contractivity and the triangle inequality.

Coherent frames require the same conjugation choice. Canonical coherent vectors obey $K|\alpha\rangle=|\bar\alpha\rangle$, so the conjugate state's kernel is
$\langle\alpha|\rho^*|\beta\rangle=
\overline{\langle\bar\alpha|\rho|\bar\beta\rangle}$.
Taking the raw imaginary part of $\langle\alpha|\rho|\beta\rangle$ does not measure standard imaginarity: the real thermal state has
\[
 \langle\alpha|\varrho_r|\beta\rangle=(1-r)
 e^{-(|\alpha|^2+|\beta|^2)/2+r\bar\alpha\beta},\quad 0<r<1,
\]
which is generally complex and changes under frame rephasing.
In the following integral, $z=(q,p)$ and $w=(q',p')$ are physical
phase-space labels, with $dz=dq\,dp$ and $dw=dq'\,dp'$.
A well-defined frame expression for the diagnostic is
$J_2=(4h^2)^{-1}\iint|\langle z|(\rho-\rho^*)|w\rangle|^2dzdw$;
the tight-frame resolution makes this exactly the Hilbert--Schmidt quantity.
The discrete expression has $h^2$ replaced by $d^2$.

\section{Continuous-frame coherence}\label{app-frame}
{\small

\noindent\textit{Proof of Theorem~\ref{continuousframe}.}

From (\ref{rhosigma}), we get \
$
\langle\alpha,\beta|\varrho_\sigma|\gamma,\delta\rangle=
2\int\limits_{\mathbb{R}^2}\mathcal{W}_{\!\varrho_\sigma} (q,p)\, \langle\alpha,\beta|\Pi(q,p)|\gamma,\delta\rangle dqdp. 
$\ \ 
By using  (\ref{coherentst}) and (\ref{Piqp}), we obtain
\[
\begin{array}{l}
\langle\alpha,\beta|\Pi(q,p)|\gamma,\delta\rangle=\sqrt{\frac{2}{h}}\ \ 
\mathrm{e}^{\frac{\pi \mathrm{i}}{h}(\alpha \beta-\gamma \delta)}
\int\limits_{-\infty}^\infty\mathrm{e}^{-\frac{2\pi \mathrm{i}}{h}\beta x}\ \ 
\mathrm{e}^{-\frac{\pi}{h}(x-\alpha)^2}\mathrm{e}^{\frac{4 \pi \mathrm{i}}{h}p(x-q)}\ \ 
\mathrm{e}^{\frac{2\pi \mathrm{i}}{h}\delta(2q-x)}\ \ \mathrm{e}^{-\frac{\pi}{h}(2q-x-\gamma)^2}dx
{}\\[4mm]
\qquad \qquad \quad   =\sqrt{\frac{2}{h}}\,
\mathrm{e}^{-\frac{\pi}{h}(\alpha^2+\gamma^2- \mathrm{i}\alpha \beta+\mathrm{i}\gamma \delta)} \ 
\mathrm{e}^{-\frac{4\pi}{h}(q^2+\mathrm{i}pq-\gamma q-\mathrm{i}\delta q)} \!\!\int\limits_{-\infty}^\infty\!\!\mathrm{e}^{-\frac{2\pi}{h}[x-\frac{1}{2}(\alpha -\mathrm{i}\beta -\gamma -\mathrm{i}\delta+2q+2\mathrm{i}p)]^2}dx \ 
\mathrm{e}^{\frac{\pi}{2h}(\alpha -\mathrm{i}\beta -\gamma -\mathrm{i}\delta+2q+2\mathrm{i}p)^2}\\[4mm]
\qquad \qquad \quad  =\sqrt{\frac{2}{h}}\,
\mathrm{e}^{-\frac{\pi}{h}(\alpha^2+\gamma^2- \mathrm{i}\alpha \beta+\mathrm{i}\gamma \delta)}\ \ 
\mathrm{e}^{-\frac{4\pi}{h}(q^2+\mathrm{i}pq-\gamma q-\mathrm{i}\delta q)}\  \sqrt{\frac{h}{2\pi}}\,\sqrt{\pi}\ 
\mathrm{e}^{\frac{\pi}{2h}(\alpha -\mathrm{i}\beta -\gamma -\mathrm{i}\delta+2q+2\mathrm{i}p)^2}\\[4mm]
\qquad \qquad \quad   =
\mathrm{e}^{-\frac{\pi}{2h}(\alpha+\gamma)^2}\ \ 
\mathrm{e}^{-\frac{\pi}{2h}(\beta+\delta)^2}\ \ 
\mathrm{e}^{-\frac{\pi\mathrm{i}}{h}(\alpha\delta-\beta\gamma)}
\ \ 
\mathrm{e}^{-\frac{2\pi}{h}[q^2-(\alpha-\mathrm{i}\beta +\gamma+\mathrm{i}\delta)q]}\ \ 
\mathrm{e}^{-\frac{2\pi}{h}[p^2-(\mathrm{i}\alpha+\beta -\mathrm{i}\gamma+\delta)p]}\\[4mm]
\end{array}
\]

In this appendix only, denote the entries of $\sigma^{-1}$ by $a,c,b$:\ $
\left(\!\!\!
\begin{array}{ll}
a & c\\
c & b
\end{array}\!\!\!\right)\!=\!
\left(\!\!\!
\begin{array}{ll}
\sigma_{11} & \!\sigma_{12}\\
\sigma_{12} & \!\sigma_{22}
\end{array}\!\!\!\right)^{\!\!-1}\!\!\!\!\!=\frac{1}{\mathrm{det}\, \sigma}
\left(\!\!\!
\begin{array}{rr}
\sigma_{22} & \!-\sigma_{12}\\
-\sigma_{12} & \!\sigma_{11}
\end{array}\!\!\!\right)
$ \  we obtain \  
$\mathcal{W}_{\!\varrho_\sigma} (q,p)\!=\!\frac{2}{h \sqrt{\det\sigma}}\, \mathrm{e}^{-\frac{2\pi}{h}(aq^2+2cqp+bp^2)}$\ \ and
\[
\begin{array}{l}
|\langle\alpha,\beta|\varrho_\sigma|\gamma,\delta\rangle|=
\left|2\int\limits_{\mathbb{R}^2}\mathcal{W}_{\!\varrho_\sigma} (q,p)\, \langle\alpha,\beta|\Pi(q,p)|\gamma,\delta\rangle dqdp\right|\\[2mm]
\qquad \qquad  =\frac{4}{h \sqrt{\det\sigma}}\, 
\mathrm{e}^{-\frac{\pi}{2h}(\alpha+\gamma)^2}\, 
\mathrm{e}^{-\frac{\pi}{2h}(\beta+\delta)^2}\, 
\left|\mathrm{e}^{-\frac{\pi\mathrm{i}}{h}(\alpha\delta-\beta\gamma)}
\int\limits_{\mathbb{R}^2}\mathrm{e}^{-\frac{2\pi}{h}[(a+1)q^2-(\alpha-\mathrm{i}\beta +\gamma+\mathrm{i}\delta-2cp)q]}dq\, 
\mathrm{e}^{-\frac{2\pi}{h}[(b+1)p^2-(\mathrm{i}\alpha+\beta -\mathrm{i}\gamma+\delta)p]}dp\right|\\[4mm]
\qquad \qquad  =\frac{4}{h \sqrt{\det\sigma}}\, 
\mathrm{e}^{-\frac{\pi}{2h}(\alpha+\gamma)^2}\, 
\mathrm{e}^{-\frac{\pi}{2h}(\beta+\delta)^2}\, 
\left|\int\limits_{\mathbb{R}^2}\mathrm{e}^{-\frac{2\pi}{h}[\sqrt{a+1}\, q-\frac{1}{2\sqrt{a+1}}(\alpha-\mathrm{i}\beta +\gamma+\mathrm{i}\delta-2cp)]^2}dq\right.\\
\qquad \qquad \qquad \qquad \qquad \qquad \qquad\qquad 
\qquad \qquad \qquad \quad \
\left.\mathrm{e}^{\frac{2\pi}{h}\frac{1}{4(a+1)}(\alpha-\mathrm{i}\beta +\gamma+\mathrm{i}\delta-2cp)^2}\ 
\mathrm{e}^{-\frac{2\pi}{h}[(b+1)p^2-(\mathrm{i}\alpha+\beta -\mathrm{i}\gamma+\delta)p]}dp\right|
\end{array}
\]
\[
\begin{array}{l}
\qquad \qquad  =\frac{4}{h \sqrt{\det\sigma}}\ 
\mathrm{e}^{-\frac{\pi}{2h}(\alpha+\gamma)^2}\ \ 
\mathrm{e}^{-\frac{\pi}{2h}(\beta+\delta)^2} \
\sqrt{\frac{h}{2\pi}}\frac{\sqrt{\pi}}{\sqrt{a+1}} \
\left| \mathrm{e}^{\frac{\pi}{2h}\frac{1}{a+1}(\alpha-\mathrm{i}\beta +\gamma+\mathrm{i}\delta)^2} \right|
\\[3mm]
\qquad \qquad \qquad \qquad \qquad \qquad 
\qquad \qquad \qquad \quad \
\left|\int\limits_{-\infty}^\infty
\mathrm{e}^{\frac{2\pi}{h}\frac{c^2}{a+1}p^2}\ \ 
\mathrm{e}^{-\frac{2\pi}{h}\frac{c}{a+1}(\alpha-\mathrm{i}\beta +\gamma+\mathrm{i}\delta)p}\ \ 
\mathrm{e}^{-\frac{2\pi}{h}[(b+1)p^2-(\mathrm{i}\alpha+\beta -\mathrm{i}\gamma+\delta)p]}dp\right|\\[2mm]
\qquad \qquad   =
2\sqrt{\frac{2}{a+1}}\frac{1}{\sqrt{h\,\det\sigma}}\, 
\mathrm{e}^{-\frac{\pi}{2h}(\alpha+\gamma)^2}
\mathrm{e}^{-\frac{\pi}{2h}(\beta+\delta)^2}
\mathrm{e}^{\frac{\pi}{2h}\frac{1}{a+1}(\alpha +\gamma)^2} 
\mathrm{e}^{-\frac{\pi}{2h}\frac{1}{a+1}(\beta -\delta)^2} 
\\
\qquad \qquad \qquad \qquad \qquad \qquad 
\qquad \qquad \qquad \quad \
\left|\int\limits_{-\infty}^\infty
\mathrm{e}^{-\frac{2\pi}{h}(b+1-\frac{c^2}{a+1})p^2+\frac{2\pi}{h}
[(\mathrm{i}\alpha+\beta -\mathrm{i}\gamma+\delta)-\frac{c}{a+1}(\alpha-\mathrm{i}\beta +\gamma+\mathrm{i}\delta)]p}
dp\right|\\[2mm]
\qquad \qquad   =
2\sqrt{\frac{2}{a+1}}\frac{1}{\sqrt{h\,\det\sigma}}\, 
\mathrm{e}^{-\frac{\pi}{2h}\frac{a}{a+1}(\alpha+\gamma)^2}
\mathrm{e}^{-\frac{\pi}{2h}(\beta+\delta)^2}
\mathrm{e}^{-\frac{\pi}{2h}\frac{1}{a+1}(\beta -\delta)^2} 
\\
\qquad \qquad \qquad \qquad\qquad \qquad\left|\int\limits_{-\infty}^\infty
\mathrm{e}^{-\frac{2\pi}{h}\left\{\sqrt{\frac{(a+1)(b+1)-c^2}{a+1}}p
-\frac{1}{2}\sqrt{\frac{a+1}{(a+1)(b+1)-c^2}}
\left[(\mathrm{i}\alpha+\beta -\mathrm{i}\gamma+\delta)-\frac{c}{a+1}(\alpha-\mathrm{i}\beta +\gamma+\mathrm{i}\delta)\right]\right\}^2}
dp\right|\\[4mm]
\qquad \qquad \qquad \qquad\qquad \qquad \qquad \qquad\qquad\qquad \qquad \qquad
\left|
\mathrm{e}^{\frac{\pi}{2h}\frac{a+1}{(a+1)(b+1)-c^2}
\left[(\mathrm{i}\alpha+\beta -\mathrm{i}\gamma+\delta)-\frac{c}{a+1}(\alpha-\mathrm{i}\beta +\gamma+\mathrm{i}\delta)\right]^2}
\right|\\[2mm]
\qquad \qquad   =
2\sqrt{\frac{2}{a+1}}\frac{1}{\sqrt{h\,\det\sigma}}\, 
\mathrm{e}^{-\frac{\pi}{2h}\frac{a}{a+1}(\alpha+\gamma)^2}
\mathrm{e}^{-\frac{\pi}{2h}(\beta+\delta)^2}
\mathrm{e}^{-\frac{\pi}{2h}\frac{1}{a+1}(\beta -\delta)^2} 
\sqrt{\frac{h}{2\pi}}\sqrt{\frac{a+1}{(a+1)(b+1)-c^2}} \, \sqrt{\pi}
\\[3mm]
\qquad \qquad \qquad \qquad\qquad \qquad \qquad \qquad \ \ 
\mathrm{e}^{-\frac{\pi}{2h}\frac{a+1}{(a+1)(b+1)-c^2}
\left[\alpha-\gamma+\frac{c}{a+1}(\beta-\delta)\right]^2} \
\mathrm{e}^{\frac{\pi}{2h}\frac{a+1}{(a+1)(b+1)-c^2}
\left[\beta +\delta-\frac{c}{a+1}(\alpha +\gamma)\right]^2}\\[2mm]
\qquad \qquad   =
\frac{2}{\sqrt{\det\sigma}}\frac{1}{\sqrt{(a+1)(b+1)-c^2}}\, 
\mathrm{e}^{-\frac{\pi}{2h}\frac{a}{a+1}(\alpha+\gamma)^2}
\mathrm{e}^{-\frac{\pi}{2h}(\beta+\delta)^2}
\mathrm{e}^{-\frac{\pi}{2h}\frac{1}{a+1}(\beta -\delta)^2} 
\\[2mm]
\qquad \qquad \qquad \qquad\qquad \qquad \qquad\qquad
\mathrm{e}^{-\frac{\pi}{2h} \frac{a+1}{(a+1)(b+1)-c^2}
\left[\alpha-\gamma+\frac{c}{a+1}(\beta-\delta)\right]^2} \
\mathrm{e}^{\frac{\pi}{2h} \frac{a+1}{(a+1)(b+1)-c^2}
\left[\beta +\delta-\frac{c}{a+1}(\alpha +\gamma)\right]^2}
\end{array}
\]

By denoting $\eta\!=\!(a\!+\!1)(b\!+\!1)-c^2$ and using the change of variables $x\!=\!\alpha\!+\!\gamma$, $y\!=\!\beta\!+\!\delta$,
$u\!=\!\alpha\!-\!\gamma$, $v\!=\!\beta\!-\!\delta$, we get
\[
\begin{array}{l}
\int\limits_{\mathbb{R}^4}|\langle \alpha,\beta|\varrho_\sigma|\gamma,\delta\rangle|d\alpha d\beta d\gamma d\delta =\\[3mm]
\qquad \qquad =\frac{1}{2\sqrt{\det\sigma}}\frac{1}{\sqrt{\eta}}\, 
\int\limits_{\mathbb{R}^4}\mathrm{e}^{-\frac{\pi}{2h}\frac{a}{a+1}x^2}
\mathrm{e}^{-\frac{\pi}{2h}y^2}
\mathrm{e}^{-\frac{\pi}{2h}\frac{1}{a+1}v^2} 
\mathrm{e}^{-\frac{\pi}{2h} \frac{a+1}{\eta}
\left[u+\frac{c}{a+1}v\right]^2} 
\mathrm{e}^{\frac{\pi}{2h} \frac{a+1}{\eta}
\left[y-\frac{c}{a+1}x\right]^2}dxdydudv\\[2mm]
\qquad \qquad =\frac{1}{2\sqrt{\det\sigma}}\frac{1}{\sqrt{\eta}}\, 
\int\limits_{\mathbb{R}^2}
\mathrm{e}^{-\frac{\pi}{2h} \frac{a+1}{\eta}
\left[u+\frac{c}{a+1}v\right]^2}\mathrm{e}^{-\frac{\pi}{2h}\frac{1}{a+1}v^2}  dudv\
\int\limits_{\mathbb{R}^2}\mathrm{e}^{\frac{\pi}{2h} \frac{a+1}{\eta}
\left[y-\frac{c}{a+1}x\right]^2}\mathrm{e}^{-\frac{\pi}{2h}\frac{a}{a+1}x^2}\mathrm{e}^{-\frac{\pi}{2h}y^2}dxdy\\[2mm]
\qquad \qquad =\frac{1}{2\sqrt{\det\sigma}}\frac{1}{\sqrt{\eta}}\, 
\sqrt{\frac{2h}{\pi}}\frac{\sqrt{\eta}}{\sqrt{a+1}} \sqrt{\pi}\ 
\sqrt{\frac{2h}{\pi}}\sqrt{a+1} \sqrt{\pi}\\[3mm]
\qquad \qquad \qquad \qquad\int\limits_{\mathbb{R}^2}\mathrm{e}^{-\frac{\pi}{2h} \left[\sqrt{\frac{ab+b-c^2}{\eta}}y+\frac{c}{\eta}
\sqrt{\frac{\eta}{ab+b-c^2}}x\right]^2}dy \ 
\mathrm{e}^{\frac{\pi}{2h}\frac{c^2}{\eta(a+1)}x^2}\mathrm{e}^{-\frac{\pi}{2h}\frac{a}{a+1}x^2}\mathrm{e}^{\frac{\pi}{2h}\frac{c^2}{\eta^2}\frac{\eta}{ab+b-c^2}x^2}dx\\[2mm]
\qquad \qquad =\frac{h}{\sqrt{\det\sigma}}
\int\limits_{-\infty}^\infty\mathrm{e}^{-\frac{\pi}{2h} \left[\sqrt{\frac{ab+b-c^2}{\eta}}y+\frac{c}{\eta}
\sqrt{\frac{\eta}{ab+b-c^2}}x\right]^2}dy \ 
\int\limits_{-\infty}^\infty
\mathrm{e}^{-\frac{\pi}{2h}\frac{ab-c^2}{ab+b-c^2}x^2}dx\\[2mm]
\qquad \qquad =\frac{h}{\sqrt{\det\sigma}}
\sqrt{\frac{2h}{\pi}}\frac{\sqrt{\eta}}{\sqrt{ab+b-c^2}} \sqrt{\pi}\ 
\int\limits_{-\infty}^\infty
\mathrm{e}^{-\frac{\pi}{2h}\frac{ab-c^2}{ab+b-c^2}x^2}dx\\[2mm]
\qquad \qquad =\frac{\sqrt{2}\,h\sqrt{h}}{\sqrt{\det\sigma}}
\frac{\sqrt{\eta}}{\sqrt{ab+b-c^2}}\ 
\int\limits_{-\infty}^\infty
\mathrm{e}^{-\frac{\pi}{2h}\frac{ab-c^2}{ab+b-c^2}x^2}dx=\frac{\sqrt{2}\,h\sqrt{h}}{\sqrt{\det\sigma}}
\frac{\sqrt{\eta}}{\sqrt{ab+b-c^2}}\ \sqrt{\frac{2h}{\pi}}\frac{\sqrt{ab+b-c^2}}{\sqrt{ab-c^2}} \sqrt{\pi}\ 
\\[2mm]
\qquad \qquad =\frac{2h^2}{\sqrt{\det\sigma}}
\frac{\sqrt{\eta}}{\sqrt{ab-c^2}}  
 =\frac{2h^2}{\sqrt{\det\sigma}}
\frac{\sqrt{ab+a+b+1-c^2}}{\sqrt{ab-c^2}}  
 =\frac{2h^2}{\sqrt{\det\sigma}}
\sqrt{ \det\sigma +\mathrm{tr}\, \sigma +1}.
\end{array}
\]

}
Multiplying this integral by $1/(2h^2)$ proves
Eq.~(\ref{coherencecsigma}). The volume elements satisfy
$d\alpha\,d\beta\,d\gamma\,d\delta=\tfrac14\,dx\,dy\,du\,dv$.


\begin{thebibliography}{99}
\bibitem{Planck}
We keep Planck constant $h=2\pi \hbar$ in our formulas in order to make the continuous-discrete correspondence more obvious.
\end{thebibliography}
\end{document}